\documentclass[11pt,american]{article}
\usepackage[T1]{fontenc}
\usepackage[utf8]{inputenc}
\usepackage{color}
\usepackage{babel}
\usepackage{mathtools}
\usepackage{amsmath}
\usepackage{amsthm}
\usepackage{amssymb}
\usepackage{graphicx}
\usepackage[pdfusetitle,
 bookmarks=true,bookmarksnumbered=false,bookmarksopen=false,
 breaklinks=false,pdfborder={0 0 0},pdfborderstyle={},backref=false,colorlinks=true]
 {hyperref}
\hypersetup{
 linkcolor=magenta, urlcolor=blue, citecolor=blue}

\makeatletter
\theoremstyle{plain}
\newtheorem{thm}{\protect\theoremname}
\theoremstyle{plain}
\newtheorem{lyxalgorithm}[thm]{\protect\algorithmname}
\theoremstyle{plain}
\newtheorem{prop}[thm]{\protect\propositionname}

\DeclareMathOperator{\Tr}{Tr}

\allowdisplaybreaks

\numberwithin{equation}{section}

\makeatother

\providecommand{\algorithmname}{Algorithm}
\providecommand{\propositionname}{Proposition}
\providecommand{\theoremname}{Theorem}

\begin{document}
\title{Generative modeling using evolved quantum Boltzmann machines}
\author{Mark M. Wilde\\
\textit{School of Electrical and Computer Engineering, Cornell University,}\\
\textit{Ithaca, New York 14850, USA}}
\maketitle
\begin{abstract}
Born-rule generative modeling, a central task in quantum machine learning,
seeks to learn probability distributions that can be efficiently sampled
by measuring complex quantum states. One hope is for quantum models
to efficiently capture probability distributions that are difficult
to learn and simulate by classical means alone. Quantum Boltzmann
machines were proposed about one decade ago for this purpose, yet
efficient training methods have remained elusive. In this paper, I
overcome this obstacle by proposing a practical solution that trains
quantum Boltzmann machines for Born-rule generative modeling. Two
key ingredients in the proposal are the Donsker--Varadhan variational
representation of the classical relative entropy and the quantum Boltzmann
gradient estimator of {[}Patel \textit{et al}., arXiv:2410.12935{]}.
I present the main result for a more general ansatz known as an evolved
quantum Boltzmann machine {[}Minervini \textit{et al}., arXiv:2501.03367{]},
which combines parameterized real- and imaginary-time evolution. I
also show how to extend the findings to other distinguishability measures
beyond relative entropy. Finally, I present four different hybrid
quantum--classical algorithms for the minimax optimization underlying
training, and I discuss their theoretical convergence guarantees.
\end{abstract}
\tableofcontents{}

\section{Introduction}

\subsection{Background and motivation}

By now, quantum machine learning is a well established field of study,
wherein one of the goals is to determine distinctions between quantum
and classical computers for learning tasks \cite{Biamonte2017,Chang2025}.
There is a rich body of theoretical work including topics as diverse
as quantum learning theory \cite{Anshu2023,Caro2024}, variational
quantum algorithms for learning \cite{Cerezo2021}, quantum algorithms
for machine learning applications \cite{Wang2024,Liu2024}, etc. We
also see ideas from the more traditional area of quantum information
theory influencing and having an impact \cite{Cheng2016,Huang2021,Caro2024}.

A key task in quantum machine learning is known as Born-rule generative
modeling \cite{Benedetti2017,Kieferova2017,PerdomoOrtiz2018,Amin2018,Coyle2020,Sweke2021,Gao2022,Rudolph2024},
in which the goal is to simulate samples from a probability distribution
$p$ by means of a parameterized probability distribution $q_{\gamma}$
that results from measuring a parameterized quantum state $\rho_{\gamma}$.
In more detail, suppose that $p$ is a probability distribution over
a finite alphabet $\mathcal{Z}$. We assume that sample access to
$p$ is available, and the model probability distribution $q_{\gamma}$
is of the following form:
\begin{equation}
q_{\gamma}(z)\coloneqq\Tr\!\left[\Lambda_{z}\rho_{\gamma}\right],
\end{equation}
where $\left(\Lambda_{z}\right)_{z\in\mathcal{Z}}$ is a positive
operator-valued measure (POVM), which satisfies $\Lambda_{z}\geq0$
for all $z\in\mathcal{Z}$ and $\sum_{z\in\mathcal{Z}}\Lambda_{z}=I$.
Additionally, $\rho_{\gamma}$ is a state (density operator) parameterized
by $\gamma\in\Gamma\subseteq\mathbb{R}^{M}$, where $M\in\mathbb{N}$.
By receiving samples from $p$ and tuning the values in the parameter
vector $\gamma$, the aim is to minimize an information-theoretically
meaningful distinguishability measure between $p$ and $q_{\gamma}$,
such as the relative entropy \cite{Kullback1951}:
\begin{equation}
D(p\|q_{\gamma})\coloneqq\sum_{z\in\mathcal{Z}}p(z)\ln\!\left(\frac{p(z)}{q_{\gamma}(z)}\right).\label{eq:rel-ent-def}
\end{equation}
That is, the goal is to determine a minimizing parameter vector $\gamma$
in the following minimization task:
\begin{equation}
\inf_{\gamma\in\Gamma}D(p\|q_{\gamma}),\label{eq:min-task-rel-ent}
\end{equation}
by observing samples from $p$, generating samples from $q_{\gamma}$,
and tuning $\gamma$. Once a suitable $\gamma$ has been determined,
one can then generate samples from $q_{\gamma}$ at will by preparing
the state $\rho_{\gamma}$ on a quantum computer and measuring it
according to the POVM $\left(\Lambda_{z}\right)_{z\in\mathcal{Z}}$.
The relative entropy is a sensible distinguishability measure, being
well grounded in information theory with an operational meaning in
asymptotic hypothesis testing as the optimal type II error exponent
in Stein's lemma \cite{Stein1951,Chernoff1956,Cover2006}. Furthermore,
via the Pinsker inequality, it is an upper bound on the total variational
distance between $p$ and $q_{\gamma}$, the latter having information-theoretic
meaning in the single-shot setting of hypothesis testing \cite[Theorem~13.1.1]{Lehmann2005}.

Several quantum models have been proposed for Born-rule generative
modeling, including quantum circuit Born machines \cite{Liu2018,Benedetti2019},
quantum Boltzmann machines \cite{Amin2018,Benedetti2017,Kieferova2017},
and a generalization of the latter known as evolved quantum Boltzmann
machines \cite{Minervini2025}. However, for the last two models,
it has been unclear since the original proposals how to train them
efficiently for Born-rule generative modeling, and there has been
wide speculation in the quantum machine learning literature that this
would not be possible \cite{Amin2018,Wiebe2019,Anschuetz2019,Kappen2020,Kieferova2017,Zoufal2021}.
Most recently, a proposal for training quantum Boltzmann machines
for Born-rule generative modeling was made in \cite{Patel2024}, but
the assumptions made in doing so seem too restrictive to be widely
applicable.

Before proceeding, let us recall that a quantum Boltzmann machine
is a variational ansatz of the following form \cite{Amin2018,Benedetti2017,Kieferova2017}:
\begin{align}
\rho_{\theta} & \coloneqq\frac{e^{-G(\theta)}}{Z(\theta)},\label{eq:QBM-def}\\
G(\theta) & \coloneqq\sum_{j=1}^{J}\theta_{j}G_{j},\\
Z(\theta) & \coloneqq\Tr\!\left[e^{-G(\theta)}\right],
\end{align}
where $\theta_{j}\in\mathbb{R}$ and each $G_{j}$ is a Hamiltonian
for all $j\in\left[J\right]\coloneqq\left\{ 1,\ldots,J\right\} $.
Thus, $\rho_{\theta}$ is a parameterized thermal state of the Hamiltonian
$G(\theta)$. Generalizing this ansatz is an evolved quantum Boltzmann
machine \cite{Minervini2025}, having the following form:
\begin{align}
\omega_{\theta,\phi} & \coloneqq e^{-iH(\phi)}\rho_{\theta}e^{iH(\phi)},\label{eq:evolved-QBM}\\
H(\phi) & \coloneqq\sum_{k=1}^{K}\phi_{k}H_{k},
\end{align}
where $\phi_{k}\in\mathbb{R}$ and each $H_{k}$ is a Hamiltonian
for all $k\in\left[K\right]$. An evolved quantum Boltzmann machine
incorporates parameterized time evolution in addition to a parameterized
thermal state.

\subsection{Summary of main contribution}

By combining two key ingredients, I propose a practical solution to
the problem of training evolved quantum Boltzmann machines for Born-rule
generative modeling.

\begin{figure}
\centering{}\includegraphics[width=0.65\textwidth]{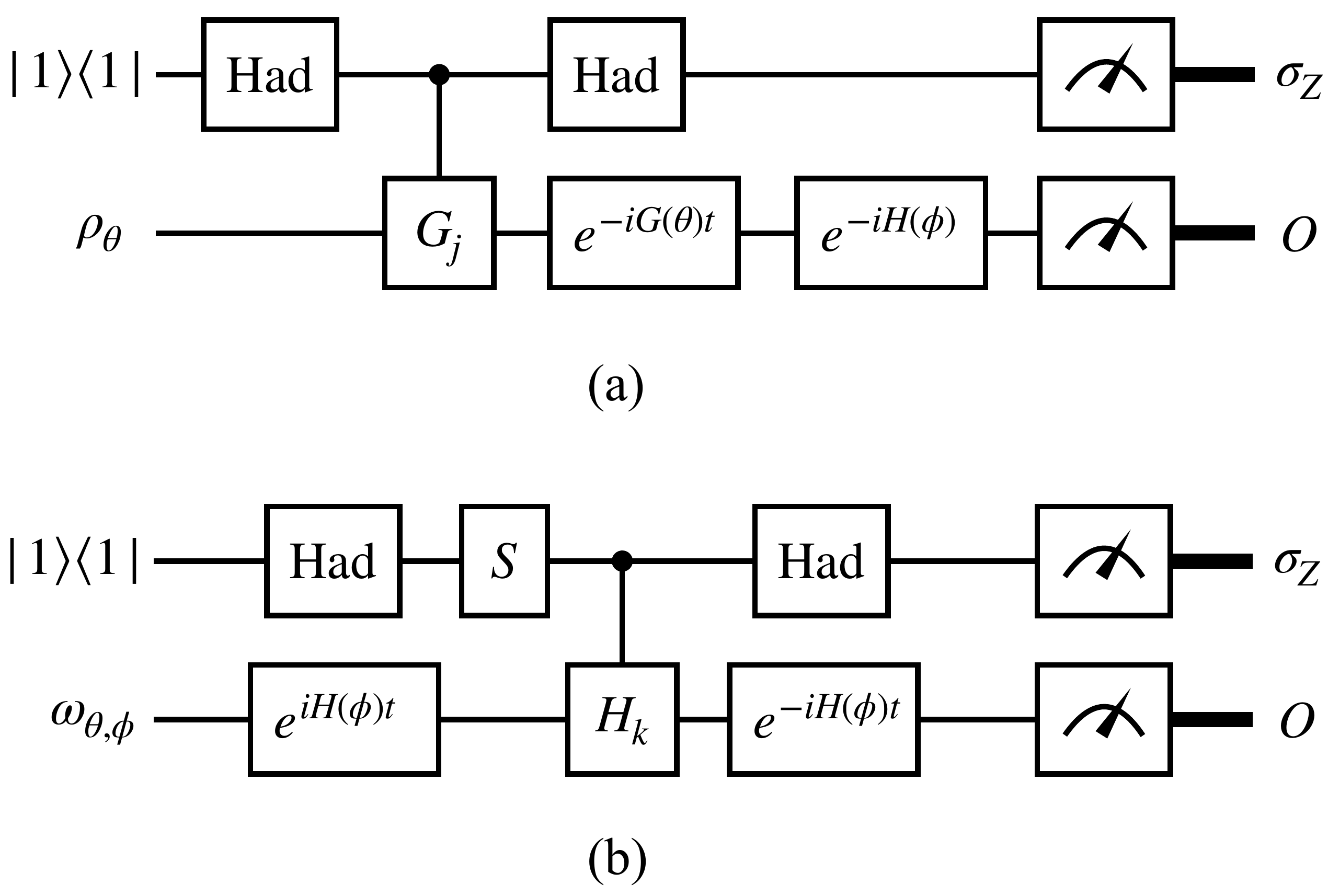}\caption{{\footnotesize (a) Quantum circuit for estimating the first term $-\frac{1}{2}\left\langle \left\{ O,\Phi_{\theta}(G_{j})\right\} \right\rangle _{\rho_{\theta}}$
in \eqref{eq:gradient-VQE-obj}. For each execution of this circuit,
a value $t\in\mathbb{R}$ is chosen randomly from the high-peak probability
density $p(t)$ in \eqref{eq:high-peak-tent-def}. (b) Quantum circuit
for estimating the quantity $-\frac{i}{2}\left\langle \left[O,\Psi_{\phi}(H_{k})\right]\right\rangle _{\omega_{\theta,\phi}}$
in \eqref{eq:time-evolve-grad}. For each execution of this circuit,
a value $t\in\mathbb{R}$ is chosen uniformly at random from the unit
interval $\left[0,1\right]$. The symbol ``Had'' denotes the Hadamard
gate, and the symbol $S$ denotes the phase gate $S\protect\coloneqq\begin{bmatrix}1 & 0\\
0 & i
\end{bmatrix}$.} }\label{fig:Quantum-Boltzmann-gradient-estimator}
\end{figure}

The first ingredient is one of the main contributions of \cite{Patel2024,Minervini2025},
known as the evolved quantum Boltzmann gradient estimator. Given an
observable $O$, the gradient of the function
\begin{equation}
\gamma\coloneqq\left(\theta,\phi\right)\mapsto\Tr\!\left[O\omega_{\theta,\phi}\right]
\end{equation}
is specified by the following partial derivatives \cite[Eqs.~(13)--(14)]{Minervini2025}:
\begin{align}
\frac{\partial}{\partial\theta_{j}}\Tr\!\left[O\omega_{\theta,\phi}\right] & =-\frac{1}{2}\left\langle \left\{ e^{iH(\phi)}Oe^{-iH(\phi)},\Phi_{\theta}(G_{j})\right\} \right\rangle _{\rho_{\theta}}+\left\langle O\right\rangle _{\omega_{\theta,\phi}}\left\langle G_{j}\right\rangle _{\rho_{\theta}},\label{eq:gradient-VQE-obj}\\
\frac{\partial}{\partial\phi_{k}}\Tr\!\left[O\omega_{\theta,\phi}\right] & =-i\left\langle \left[O,\Psi_{\phi}(H_{k})\right]\right\rangle _{\omega_{\theta,\phi}},\label{eq:time-evolve-grad}
\end{align}
where $\left\langle A\right\rangle _{\rho_{\theta}}\equiv\Tr\!\left[A\rho_{\theta}\right]$,
$\left\{ A,B\right\} \equiv AB+BA$, $\left[A,B\right]\equiv AB-BA$,
$\Phi_{\theta}$ and $\Psi_{\phi}$ are the following quantum channels:
\begin{align}
\Phi_{\theta}(X) & \coloneqq\int_{-\infty}^{\infty}dt\ p(t)\,e^{-iG(\theta)t}Xe^{iG(\theta)t},\\
\Psi_{\phi}(X) & \coloneqq\int_{0}^{1}dt\ e^{-iH(\phi)t}Xe^{iH(\phi)t},\\
p(t) & \coloneqq\frac{2}{\pi}\ln\left|\coth\!\left(\frac{\pi t}{2}\right)\right|,\label{eq:high-peak-tent-def}
\end{align}
and $p(t)$ is a probability density function known as the high-peak
tent probability density \cite{Patel2024}. The quantum Boltzmann
gradient estimator is a quantum algorithm that estimates the expression
in \eqref{eq:gradient-VQE-obj} for each $j\in\left[J\right]$. An
assumption of these gradient estimation algorithms is that one can
efficiently prepare the states $\omega_{\theta,\phi}$ and $\rho_{\theta}$,
which is justified by efficient quantum algorithms for Hamiltonian
simulation \cite{Lloyd1996,Childs2018} and thermal state preparation
\cite{Chen2025,Rouze2024,Ding2025}. The second term in \eqref{eq:gradient-VQE-obj}
is easy to estimate by repeatedly preparing the state $\omega_{\theta,\phi}\otimes\rho_{\theta}$
and measuring the observable $O\otimes G_{j}$, while the first term
can be estimated by the quantum circuit depicted in Figure \ref{fig:Quantum-Boltzmann-gradient-estimator}(a).
The quantum circuit depicted in Figure \ref{fig:Quantum-Boltzmann-gradient-estimator}(b)
estimates the partial derivative in \eqref{eq:time-evolve-grad}.
Altogether, the scheme for estimating all partial derivatives in \eqref{eq:gradient-VQE-obj}--\eqref{eq:time-evolve-grad}
is known as the evolved quantum Boltzmann gradient estimator. As given
in this form, both algorithms assume that $G_{j}$ and $H_{k}$ are
unitary in addition to being Hermitian; however, as stated in \cite{Patel2024,Minervini2025},
this assumption can be removed if block-encodings \cite{Low2019,Gilyen2019}
of $G_{j}$ and $H_{k}$ are available. Henceforth, for simplicity,
I make the assumption that $G_{j}$ and $H_{k}$ are both Hermitian
and unitary.

The second key ingredient is a theoretical device known as the Donsker--Varadhan
(DV) variational formula \cite[Lemma~2.1]{Donsker1975}, which has
been widely used and studied in the classical generative modeling
literature \cite{Nguyen2010,Nowozin2016,Belghazi2018,Sreekumar2022}.
In short, the DV formula provides the following variational representations
for the relative entropy in \eqref{eq:rel-ent-def}:
\begin{align}
D(p\|q_{\gamma}) & =\sup_{T\in\mathcal{T}}\left\{ \mathbb{E}_{p}\!\left[T(Z)\right]-\ln\mathbb{E}_{q_{\gamma}}\!\left[e^{T(Z)}\right]\right\} \\
 & =\sup_{T\in\mathcal{T}}\left\{ \mathbb{E}_{p}\!\left[T(Z)\right]+1-\mathbb{E}_{q_{\gamma}}\!\left[e^{T(Z)}\right]\right\} ,\label{eq:DV-formula}
\end{align}
where $\mathcal{T}$ consists of all functions from $\mathcal{Z}$
to $\mathbb{R}$ (i.e., $T\colon\mathcal{Z}\to\mathbb{R}$). Thus,
the desired optimization task in \eqref{eq:min-task-rel-ent} can
be reformulated as the following minimax optimization problem:
\begin{equation}
\inf_{\gamma\in\Gamma}D(p\|q_{\gamma})=\inf_{\gamma\in\Gamma}\sup_{T\in\mathcal{T}}\left\{ \mathbb{E}_{p}\!\left[T(Z)\right]+1-\mathbb{E}_{q_{\gamma}}\!\left[e^{T(Z)}\right]\right\} .\label{eq:min-max-opt-rel-ent}
\end{equation}
The main advantage of the formula in \eqref{eq:min-max-opt-rel-ent}
is that it rewrites the non-linear objective function in \eqref{eq:rel-ent-def}
in terms of the linear objective function in \eqref{eq:min-max-opt-rel-ent},
at the cost of an additional optimization. However, the payoff of
this rewriting is notable and useful: due to the objective function
being linear in $p$ and $q_{\gamma}$, terms like $\mathbb{E}_{p}\!\left[T(Z)\right]$
and $\mathbb{E}_{q_{\gamma}}\!\left[e^{T(Z)}\right]$ can be estimated
in an unbiased way by sampling from the probability distributions
$p$ and $q_{\gamma}$. Furthermore, the objective function in \eqref{eq:min-max-opt-rel-ent}
is concave in $T$, which is beneficial when considering schemes for
performing the minimax optimization in \eqref{eq:min-max-opt-rel-ent}.

Following earlier works \cite{Nguyen2010,Nowozin2016,Belghazi2018,Sreekumar2022},
instead of optimizing over every possible function $T\in\mathcal{T}$,
we can instead adopt a neural network parameterization of such functions
in terms of a parameter vector $w\in\mathcal{W}\subseteq\mathbb{R}^{L}$,
where $L\in\mathbb{N}$, leading to the following lower bound on $\inf_{\gamma\in\Gamma}D(p\|q_{\gamma})$:
\begin{align}
\inf_{\gamma\in\Gamma}D(p\|q_{\gamma}) & \geq\inf_{\gamma\in\Gamma}\sup_{w\in\mathcal{W}}f(\gamma,w),\label{eq:neural-net-param-div}\\
\text{where}\qquad f(\gamma,w) & \coloneqq\mathbb{E}_{p}\!\left[T_{w}(Z)\right]+1-\mathbb{E}_{q_{\gamma}}\!\left[e^{T_{w}(Z)}\right].\label{eq:neural-net-param-div-obj-func}
\end{align}
The idea from here is to optimize the objective function in \eqref{eq:neural-net-param-div-obj-func}
in terms of $\gamma$ and $w$ by using various minimax optimization
algorithms \cite{Korpelevich1976,Wang2020,Lin2020,Gorbunov2022,Gao2022a,Lin2024},
for which, in some cases, there are guarantees on their convergence
to a local minimax point \cite{Jin2020} of $f(\gamma,w)$. The optimization
task involves a minimization of $f(\gamma,w)$ with respect to $\gamma$
and a maximization with respect to $w$. The minimax optimization
algorithms also require access to the gradient of $f(\gamma,w)$ and
some of them to its Hessian, and one key insight of the present paper
is that the required gradient and Hessian elements can be estimated
by means of the evolved quantum Boltzmann gradient estimator, using
the quantum circuits depicted in Figure \ref{fig:Quantum-Boltzmann-gradient-estimator}.

Thus, the solution proposed here is a practical approach to training
evolved quantum Boltzmann machines for Born-rule generative modeling.
The key aspect of it is that the objective function $f(\gamma,w)$,
its gradient, and its Hessian can be efficiently estimated by means
of quantum algorithms, and the related minimax optimization can be
carried out as a hybrid quantum--classical algorithm. Other distinguishability
measures beyond the relative entropy, including the R\'enyi relative
quasi-entropy, the chi-square divergence, the Hellinger distance,
etc.~have variational representations similar to that in \eqref{eq:DV-formula},
so that a similar approach can be applied to them. I show this explicitly
in Section \ref{sec:Alternative-dist-meas} for the R\'enyi relative
quasi-entropy.

\subsection{Paper organization}

The rest of this paper is organized as follows:
\begin{itemize}
\item Section \ref{sec:Review-EQBGE} reviews the evolved quantum Boltzmann
gradient estimator in more detail, specifically indicating the number
of quantum circuit executions that are needed for reliable estimation
of the gradient components in \eqref{eq:gradient-VQE-obj} and \eqref{eq:time-evolve-grad}.
\item Section \ref{sec:Nonconvexity-GM-EQBM} proves that the optimization
problem in \eqref{eq:min-task-rel-ent} is not convex in the parameter
vector $\gamma$, implying that finding a local minimum is the best
we can generally attempt to accomplish for this problem.
\item Section \ref{sec:Gradient-and-Hessian-DV} presents the core technical
contribution of this paper: analytical expressions for the gradient
and Hessian of the objective function $f(\gamma,w)$ in \eqref{eq:neural-net-param-div-obj-func}
and the observation that they can be estimated by means of the evolved
quantum Boltzmann gradient estimator depicted in Figure \ref{fig:Quantum-Boltzmann-gradient-estimator}.
Section \ref{sec:Gradient-and-Hessian-DV} also presents simple bounds
on the norm of the gradients and the spectral norm of the Hessians,
which can be useful for determining the number of steps required for
minimax optimization.
\item Section \ref{sec:Hybrid-quantum=002013classical-algs} appeals directly
to existing algorithms for minimax optimization in order to formulate
four different hybrid quantum--classical algorithms for optimizing
\eqref{eq:neural-net-param-div} (Algorithms \ref{alg:Extragradient}--\ref{alg:HessianFR}).
The existing algorithms for minimax optimization include extragradient
\cite{Korpelevich1976,Gorbunov2022}, two-timescale gradient descent-ascent
\cite{Lin2020,Lin2024}, follow-the-ridge \cite{Wang2020}, and HessianFR
\cite{Gao2022a}.
\item Section \ref{sec:Maintaining-concavity-via-MLFS} specializes the
development from a general neural network to a linear model in feature
space, with the main advantage of doing so being that concavity of
\eqref{eq:neural-net-param-div-obj-func} in $w$ holds for this model,
so that stronger guarantees on convergence to a local minimax point
can be made for this model.
\item Section \ref{sec:Alternative-dist-meas} extends the whole development
to other distinguishability measures beyond relative entropy, which
includes the R\'enyi relative quasi-entropies. The variational representation
of the R\'enyi relative quasi-entropies is somewhat similar to the
DV formula in \eqref{eq:DV-formula}, so that the extension is straightforward.
\item Section \ref{sec:Conclusion} concludes with a brief summary of the
findings reported here, as well as suggestions for future directions.
\end{itemize}

\section{Review of evolved quantum Boltzmann gradient estimator}

\label{sec:Review-EQBGE}In this section, I review the evolved quantum
Boltzmann gradient estimator of \cite{Patel2024,Minervini2025}. In
particular, I recall how the quantum circuits depicted in Figure \ref{fig:Quantum-Boltzmann-gradient-estimator}
can be used to estimate the first term of \eqref{eq:gradient-VQE-obj}
and the expression in \eqref{eq:time-evolve-grad}.

Suppose that $O$ is an observable of the following form:
\begin{equation}
O=\sum_{z\in\mathcal{Z}}g(z)\Lambda_{z},
\end{equation}
where $g\colon\mathcal{Z}\to\mathbb{R}$ and $\left(\Lambda_{z}\right)_{z\in\mathcal{Z}}$
is a POVM. It follows that
\begin{equation}
\left\Vert O\right\Vert \leq\left\Vert g\right\Vert \coloneqq\max_{z\in\mathcal{Z}}\left|g(z)\right|,\label{eq:observable-norm-bound}
\end{equation}
where $\left\Vert O\right\Vert \coloneqq\sup_{|\psi\rangle\neq0}\frac{\left\Vert O|\psi\rangle\right\Vert }{\left\Vert |\psi\rangle\right\Vert }$
is the spectral norm, because
\begin{equation}
-\left\Vert g\right\Vert I=-\left\Vert g\right\Vert \sum_{z\in\mathcal{Z}}\Lambda_{z}\leq\sum_{z\in\mathcal{Z}}g(z)\Lambda_{z}=O\leq\left\Vert g\right\Vert \sum_{z\in\mathcal{Z}}\Lambda_{z}=\left\Vert g\right\Vert I,
\end{equation}
which is helpful for determining the number of samples sufficient
for accurate gradient estimation.

The first algorithm estimates the first term of \eqref{eq:gradient-VQE-obj}:
\begin{lyxalgorithm}
\label{alg:EQBGE-1}The algorithm for estimating
\begin{equation}
\mu\equiv-\frac{1}{2}\left\langle \left\{ e^{iH(\phi)}Oe^{-iH(\phi)},\Phi_{\theta}(G_{j})\right\} \right\rangle _{\rho_{\theta}}
\end{equation}
proceeds according to the following steps:
\begin{enumerate}
\item Fix the desired accuracy $\varepsilon>0$ and error probability $\delta\in\left(0,1\right)$.
Set $n\in\mathbb{N}$ to satisfy $n\geq\frac{\left\Vert g\right\Vert ^{2}}{2\varepsilon^{2}}\ln\!\left(\frac{2}{\delta}\right)$.
\item For $i=1,\ldots,n$, execute the quantum circuit depicted in Figure
\ref{fig:Quantum-Boltzmann-gradient-estimator}(a). Set $X_{i}\leftarrow\left(-1\right)^{r_{i}}g(z_{i})$,
where $r_{i}\in\left\{ 0,1\right\} $ is the outcome of the $\sigma_{Z}$
measurement on the control register and $z_{i}$ is the outcome of
the measurement $\left(\Lambda_{z}\right)_{z\in\mathcal{Z}}$ on the
data register.
\item Output $\overline{X}\equiv\frac{1}{n}\sum_{i=1}^{n}X_{i}$ as an estimate
of $\mu$.
\end{enumerate}
\end{lyxalgorithm}

The second algorithm estimates the expression in \eqref{eq:time-evolve-grad}.
\begin{lyxalgorithm}
\label{alg:EQBGE-2}The algorithm for estimating
\begin{equation}
\nu\equiv-\frac{i}{2}\left\langle \left[O,\Psi_{\phi}(H_{k})\right]\right\rangle _{\omega_{\theta,\phi}}
\end{equation}
proceeds according to the following steps:
\begin{enumerate}
\item Fix the desired accuracy $\varepsilon>0$ and error probability $\delta\in\left(0,1\right)$.
Set $n\in\mathbb{N}$ to satisfy $n\geq\frac{2\left\Vert g\right\Vert ^{2}}{\varepsilon^{2}}\ln\!\left(\frac{2}{\delta}\right)$.
\item For $i=1,\ldots,n$, execute the quantum circuit depicted in Figure
\ref{fig:Quantum-Boltzmann-gradient-estimator}(b). Set $Y_{i}\leftarrow2\left(-1\right)^{s_{i}}g(z'_{i})$,
where $s_{i}\in\left\{ 0,1\right\} $ is the outcome of the $\sigma_{Z}$
measurement on the control register and $z'_{i}$ is the outcome of
the measurement $\left(\Lambda_{z}\right)_{z\in\mathcal{Z}}$ on the
data register.
\item Output $\overline{Y}\equiv\frac{1}{n}\sum_{i=1}^{n}Y_{i}$ as an estimate
of $\nu$.
\end{enumerate}
\end{lyxalgorithm}

The analysis in \cite[App.~B \& C]{Minervini2025} demonstrates that
$\overline{X}$ and $\overline{Y}$ are unbiased estimators of $\mu$
and $\nu$ because the following equalities hold for all $i$:
\begin{align}
\mathbb{E}\!\left[\left(-1\right)^{R_{i}}g(Z_{i})\right] & =\mu,\\
\mathbb{E}\!\left[2\left(-1\right)^{S_{i}}g(Z'_{i})\right] & =\nu.
\end{align}
According to the Hoeffding inequality \cite{Hoeffding1963} (see also,
e.g., \cite[Lemma~18]{Rethinasamy2023}), the following inequalities
hold:
\begin{align}
\Pr\!\left[\left|\overline{X}-\mu\right|\leq\varepsilon\right] & \geq1-\delta,\\
\Pr\!\left[\left|\overline{Y}-\nu\right|\leq\varepsilon\right] & \geq1-\delta,
\end{align}
indicating that Algorithms \ref{alg:EQBGE-1} and \ref{alg:EQBGE-2}
provide accurate estimations of $\mu$ and $\nu$ with sufficiently
many quantum circuit executions.

\section{Nonconvexity of Born-rule generative modeling with evolved quantum
Boltzmann machines}

\label{sec:Nonconvexity-GM-EQBM}In this section, I point out that
the optimization problem in \eqref{eq:min-task-rel-ent} is not convex
in the parameter vector $\gamma$ when the underlying ansatz is an
evolved quantum Boltzmann machine, as in \eqref{eq:evolved-QBM}.
In fact, we can arrive at this conclusion simply by considering a
quantum Boltzmann machine consisting of a single qubit. As a consequence
of this finding, the best we can hope for, in general, when attempting
to perform the optimization in \eqref{eq:min-task-rel-ent} is to
find a local minimum. Thus, this finding stands in distinction to
the quantum state learning problem from \cite{Coopmans2024}, where
it was proven that the optimization problem considered there is convex.

\begin{figure}
\begin{centering}
\includegraphics[width=0.8\textwidth]{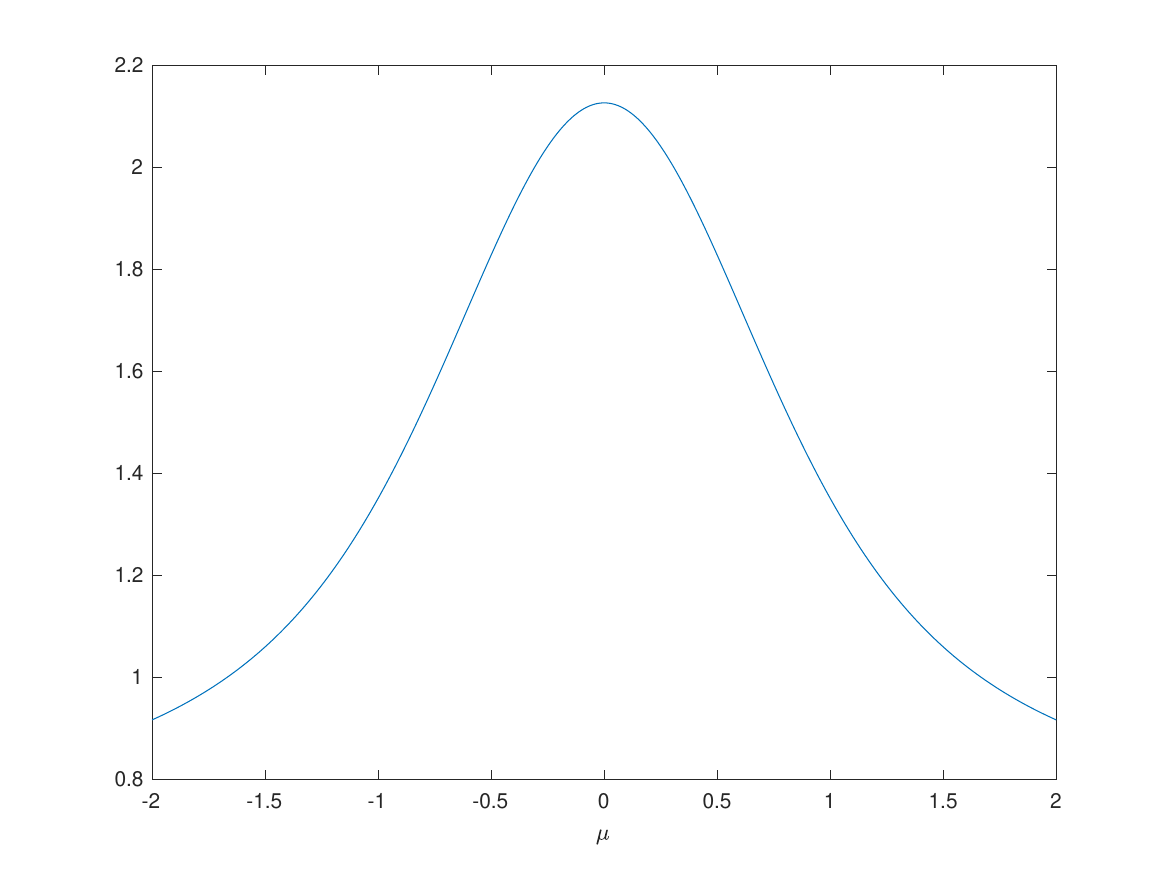}
\par\end{centering}
\caption{Plot of the function in \eqref{eq:rel-ent-not-convex-example}, demonstrating
that it is not convex in $\mu$.}\label{fig:not-convex-example}

\end{figure}

The counterexample to convexity is as follows. Let $G(\theta)=\theta_{X}\sigma_{X}+\theta_{Z}\sigma_{Z}$,
where $\sigma_{X}$ and $\sigma_{Z}$ are the standard Pauli matrices
and $\theta=\left(\theta_{X},\theta_{Z}\right)\in\mathbb{R}$. The
parameterized thermal state $\rho_{\theta}$ is then defined as in
\eqref{eq:QBM-def}. In this case, the thermal state has a closed
form as follows (see, e.g., \cite[Corollary~17]{Minervini2025a}):
\begin{equation}
\rho_{\theta}=\frac{1}{2}\left(I+r_{X}\sigma_{X}+r_{Z}\sigma_{Z}\right),
\end{equation}
where
\begin{equation}
r_{X}\coloneqq-\frac{\tanh(r)}{r}\theta_{X},\qquad r_{Z}\coloneqq-\frac{\tanh(r)}{r}\theta_{Z},\qquad r\coloneqq\sqrt{\theta_{X}^{2}+\theta_{Z}^{2}}.
\end{equation}
Measuring this state in the standard basis $\left\{ |0\rangle,|1\rangle\right\} $
then leads to the following model probability distribution:
\begin{equation}
q_{\theta}(0)=\langle0|\rho_{\theta}|0\rangle=\frac{1}{2}\left(1-\frac{\tanh(r)}{r}\theta_{Z}\right),\qquad q_{\theta}(1)=1-q_{\theta}(0).
\end{equation}
We can take the probability distribution $p$ to be $p(0)=1$ and
$p(1)=0$. The relative entropy then evaluates to
\begin{equation}
D(p\|q_{\theta})=-\ln q_{\theta}(0)=-\ln\!\left[\frac{1}{2}\left(1-\frac{\tanh(r)}{r}\theta_{Z}\right)\right].
\end{equation}
Now let us set $\theta_{X}=\mu$ and $\theta_{Z}=1$, so that the
relative entropy reduces to
\begin{equation}
-\ln\!\left[\frac{1}{2}\left(1-\frac{\tanh\!\left(\sqrt{\mu^{2}+1}\right)}{\sqrt{\mu^{2}+1}}\right)\right].\label{eq:rel-ent-not-convex-example}
\end{equation}
Figure \ref{fig:not-convex-example} plots this function, demonstrating
that it is not convex in $\mu$.

\section{Gradient and Hessian of Donsker--Varadhan objective function}

\label{sec:Gradient-and-Hessian-DV}In order to perform the optimization
in \eqref{eq:neural-net-param-div} (i.e., minimax optimization of
the Donsker--Varadhan objective function), it is helpful to determine
the gradient and Hessian of the objective function $f(\gamma,w)$
in \eqref{eq:neural-net-param-div-obj-func}. Indeed, all of the hybrid
quantum--classical algorithms delineated in Section \ref{sec:Hybrid-quantum=002013classical-algs}
make use of the gradient, and some of them also make use of the Hessian.

\subsection{Analytical expressions for gradient and Hessian and quantum algorithms
for estimating them}

\label{subsec:Analytical-expressions-grad-Hess}To this end, in this
subsection, I derive analytical expressions for the gradient and Hessian
of $f(\gamma,w)$, and I also outline how some of their elements can
be estimated by means of Algorithms \ref{alg:EQBGE-1} and \ref{alg:EQBGE-2}.

To begin with, let us first rewrite $f(\gamma,w)$ in \eqref{eq:neural-net-param-div-obj-func}
as follows:
\begin{align}
f(\gamma,w) & =\sum_{z\in\mathcal{Z}}p(z)T_{w}(z)+1-\left\langle O_{w}\right\rangle _{\omega_{\theta,\phi}},\label{eq:obj-func-rewrite}
\end{align}
where $\gamma=\left(\theta,\phi\right)$, the observable $O_{w}$
is defined as
\begin{equation}
O_{w}\coloneqq\sum_{z\in\mathcal{Z}}e^{T_{w}(z)}\Lambda_{z},
\end{equation}
and we used that
\begin{align}
\mathbb{E}_{p}\!\left[T_{w}(Z)\right] & =\sum_{z\in\mathcal{Z}}p(z)T_{w}(z),\\
\mathbb{E}_{q_{\gamma}}\!\left[e^{T_{w}(Z)}\right] & =\sum_{z\in\mathcal{Z}}q_{\gamma}(z)e^{T_{w}(z)}\label{eq:rewrite-as-obs-O-1}\\
 & =\sum_{z\in\mathcal{Z}}\Tr\!\left[\Lambda_{z}\omega_{\theta,\phi}\right]e^{T_{w}(z)}\\
 & =\Tr\!\left[O_{w}\omega_{\theta,\phi}\right]\label{eq:rewrite-as-obs-O-mid}\\
 & =\left\langle O_{w}\right\rangle _{\omega_{\theta,\phi}}.\label{eq:rewrite-as-obs-O-last}
\end{align}

With \eqref{eq:obj-func-rewrite} in hand, it follows as a direct
consequence of \eqref{eq:gradient-VQE-obj}--\eqref{eq:time-evolve-grad}
that the gradient $\nabla_{\gamma}f(\gamma,w)$ with respect to $\gamma=\left(\theta,\phi\right)$
has the following components:
\begin{align}
\frac{\partial}{\partial\theta_{j}}f(\gamma,w) & =\frac{\partial}{\partial\theta_{j}}\left(-\left\langle O_{w}\right\rangle _{\omega_{\theta,\phi}}\right)\\
 & =\frac{1}{2}\left\langle \left\{ e^{iH(\phi)}O_{w}e^{-iH(\phi)},\Phi_{\theta}(G_{j})\right\} \right\rangle _{\rho_{\theta}}-\left\langle O_{w}\right\rangle _{\omega_{\theta,\phi}}\left\langle G_{j}\right\rangle _{\rho_{\theta}}\label{eq:gradient-gamma-theta}\\
\frac{\partial}{\partial\phi_{k}}f(\gamma,w) & =\frac{\partial}{\partial\phi_{k}}\left(-\left\langle O_{w}\right\rangle _{\omega_{\theta,\phi}}\right)\\
 & =i\left\langle \left[O_{w},\Psi_{\phi}(H_{k})\right]\right\rangle _{\omega_{\theta,\phi}}.\label{eq:gradient-gamma-phi}
\end{align}
Thus, it is possible to estimate the gradient $\nabla_{\gamma}f(\gamma,w)$
by means of Algorithms \ref{alg:EQBGE-1} and \ref{alg:EQBGE-2},
with $O$ therein replaced by $O_{w}$. In addition, the number of
samples sufficient for obtaining $\varepsilon$-accurate estimates
with error probability $\leq\delta$ is given by $\mathcal{O}\!\left(\frac{\left\Vert e^{T_{w}}\right\Vert ^{2}}{\varepsilon^{2}}\ln\!\left(\frac{1}{\delta}\right)\right)$.

The gradient of $f(\gamma,w)$ with respect to the neural-network
parameter vector $w\in\mathcal{W}$ is as follows:
\begin{equation}
\nabla_{w}f(\gamma,w)=\mathbb{E}_{p}\!\left[\nabla_{w}T_{w}(Z)\right]-\mathbb{E}_{q_{\gamma}}\!\left[e^{T_{w}(Z)}\nabla_{w}T_{w}(Z)\right],\label{eq:gradient-w}
\end{equation}
and the Hessian is given by
\begin{align}
H_{ww} & \coloneqq\nabla_{w}^{2}f(\gamma,w)\label{eq:hessian-w-w-alt}\\
 & =\mathbb{E}_{p}\!\left[\nabla_{w}^{2}T_{w}(Z)\right]-\mathbb{E}_{q_{\gamma}}\!\left[e^{T_{w}(Z)}\left(\nabla_{w}^{2}T_{w}(Z)+\nabla_{w}T_{w}(Z)\left(\nabla_{w}T_{w}(Z)\right)^{T}\right)\right]\label{eq:hessian-w-w}
\end{align}
Both the gradient in \eqref{eq:gradient-w} and the Hessian in \eqref{eq:hessian-w-w}
can be estimated by sampling from $p$ and $q_{\gamma}$ and applying
the standard backpropagation algorithm.

By following a development similar to that in \eqref{eq:rewrite-as-obs-O-1}--\eqref{eq:rewrite-as-obs-O-last},
we can rewrite the components of the gradient $\nabla_{w}f(\gamma,w)$
in \eqref{eq:gradient-w} as follows:
\begin{equation}
\frac{\partial}{\partial w_{\ell}}f(\gamma,w)=\mathbb{E}_{p}\!\left[\frac{\partial}{\partial w_{\ell}}T_{w}(Z)\right]-\left\langle P_{w_{\ell}}\right\rangle _{\omega_{\theta,\phi}},\label{eq:gradient-w-DV}
\end{equation}
where the observable $P_{w_{\ell}}$ is defined as
\begin{equation}
P_{w_{\ell}}\coloneqq\sum_{z\in\mathcal{Z}}e^{T_{w}(z)}\frac{\partial}{\partial w_{\ell}}T_{w}(z)\Lambda_{z}.
\end{equation}
It then follows as a direct consequence of \eqref{eq:gradient-VQE-obj}--\eqref{eq:time-evolve-grad}
that the Hessian elements of $f(\gamma,w)$ with respect to $\gamma=\left(\theta,\phi\right)$
and $w$, denoted collectively by $H_{w\gamma}$, are as follows:
\begin{align}
\frac{\partial}{\partial\theta_{j}}\frac{\partial}{\partial w_{\ell}}f(\gamma,w) & =\frac{\partial}{\partial\theta_{j}}\left(-\left\langle P_{w_{\ell}}\right\rangle _{\omega_{\theta,\phi}}\right)\\
 & =\frac{1}{2}\left\langle \left\{ e^{iH(\phi)}P_{w_{\ell}}e^{-iH(\phi)},\Phi_{\theta}(G_{j})\right\} \right\rangle _{\rho_{\theta}}-\left\langle P_{w_{\ell}}\right\rangle _{\omega_{\theta,\phi}}\left\langle G_{j}\right\rangle _{\rho_{\theta}},\label{eq:hessian-w-gamma-1}\\
\frac{\partial}{\partial\phi_{k}}\frac{\partial}{\partial w_{\ell}}f(\gamma,w) & =\frac{\partial}{\partial\phi_{k}}\left(-\left\langle P_{w_{\ell}}\right\rangle _{\omega_{\theta,\phi}}\right)\\
 & =i\left\langle \left[P_{w_{\ell}},\Psi_{\phi}(H_{k})\right]\right\rangle _{\omega_{\theta,\phi}},\label{eq:hessian-w-gamma-2}
\end{align}
Thus, each of the terms in \eqref{eq:gradient-w-DV}, \eqref{eq:hessian-w-gamma-1},
and \eqref{eq:hessian-w-gamma-2} can again be estimated by means
of Algorithms \ref{alg:EQBGE-1} and \ref{alg:EQBGE-2}, with $O$
therein replaced by $P_{w_{\ell}}$. In addition, the number of samples
sufficient for obtaining $\varepsilon$-accurate estimates with error
probability $\leq\delta$ is given by $\mathcal{O}\!\left(\frac{\left\Vert h_{\ell}\right\Vert ^{2}}{\varepsilon^{2}}\ln\!\left(\frac{1}{\delta}\right)\right)$,
where $h_{\ell}(z)\equiv e^{T_{w}(z)}\frac{\partial}{\partial w_{\ell}}T_{w}(z)$.

\subsection{Bounds on gradient and Hessian norms}

In this subsection, I delineate bounds on the norms of the gradient
and Hessian elements, which are helpful for establishing smoothness
properties of the optimization problem and give a sense of how many
iterations of gradient-based algorithms are required when using estimates
of these quantities.
\begin{prop}[Gradient bounds]
\label{prop:gradient-bounds}The following inequalities hold:
\begin{align}
\left\Vert \nabla_{\gamma}f(\gamma,w)\right\Vert ^{2} & \leq4\left\Vert e^{T_{w}}\right\Vert ^{2}\left(\sum_{j\in\left[J\right]}\left\Vert G_{j}\right\Vert ^{2}+\sum_{k\in\left[K\right]}\left\Vert H_{k}\right\Vert ^{2}\right),\label{eq:gradient-gamma-bound}\\
\left\Vert \nabla_{w}f(\gamma,w)\right\Vert ^{2} & \leq\left(\left\Vert e^{T_{w}}\right\Vert +1\right)^{2}\sum_{\ell\in\left[L\right]}\left\Vert \frac{\partial}{\partial w_{\ell}}T_{w}\right\Vert ^{2}.
\end{align}
\end{prop}

\begin{proof}
Consider that
\begin{align}
 & \left|\frac{\partial}{\partial\theta_{j}}f(\gamma,w)\right|\nonumber \\
 & =\left|\frac{1}{2}\left\langle \left\{ e^{iH(\phi)}O_{w}e^{-iH(\phi)},\Phi_{\theta}(G_{j})\right\} \right\rangle _{\rho_{\theta}}-\left\langle O_{w}\right\rangle _{\omega_{\theta,\phi}}\left\langle G_{j}\right\rangle _{\rho_{\theta}}\right|\\
 & \leq\frac{1}{2}\left|\left\langle \left\{ e^{iH(\phi)}O_{w}e^{-iH(\phi)},\Phi_{\theta}(G_{j})\right\} \right\rangle _{\rho_{\theta}}\right|+\left|\left\langle O_{w}\right\rangle _{\omega_{\theta,\phi}}\left\langle G_{j}\right\rangle _{\rho_{\theta}}\right|\\
 & \leq\left\Vert e^{iH(\phi)}O_{w}e^{-iH(\phi)}\right\Vert \left\Vert \Phi_{\theta}(G_{j})\right\Vert +\left\Vert O_{w}\right\Vert \left\Vert G_{j}\right\Vert \\
 & \leq2\left\Vert O_{w}\right\Vert \left\Vert G_{j}\right\Vert \\
 & \leq2\left\Vert e^{T_{w}}\right\Vert \left\Vert G_{j}\right\Vert ,
\end{align}
where the second inequality follows from the H\"older inequality
(specifically, $\left|\Tr\!\left[AB\right]\right|\leq\left\Vert A\right\Vert \left\Vert B\right\Vert _{1}$),
the third inequality follows from convexity and unitary invariance
of the spectral norm, and the last inequality follows from \eqref{eq:observable-norm-bound}.
Similarly, consider that
\begin{align}
\left|\frac{\partial}{\partial\phi_{k}}f(\gamma,w)\right| & =\left|i\left\langle \left[O_{w},\Psi_{\phi}(H_{k})\right]\right\rangle _{\omega_{\theta,\phi}}\right|\\
 & \leq2\left\Vert O_{w}\right\Vert \left\Vert \Psi_{\phi}(H_{k})\right\Vert \\
 & \leq2\left\Vert e^{T_{w}}\right\Vert \left\Vert H_{k}\right\Vert .
\end{align}
The proof of \eqref{eq:gradient-gamma-bound} is concluded by combining
the bounds above and noting that the components of $\nabla_{\gamma}f(\gamma,w)$
are given by $\frac{\partial}{\partial\theta_{j}}f(\gamma,w)$ for
all $j\in\left[J\right]$ and $\frac{\partial}{\partial\phi_{k}}f(\gamma,w)$
for all $k\in\left[K\right]$.

Now consider that
\begin{align}
\left|\frac{\partial}{\partial w_{\ell}}f(\gamma,w)\right| & =\left|\mathbb{E}_{p}\!\left[\frac{\partial}{\partial w_{\ell}}T_{w}(Z)\right]-\left\langle \sum_{z\in\mathcal{Z}}e^{T_{w}(z)}\frac{\partial}{\partial w_{\ell}}T_{w}(z)\Lambda_{z}\right\rangle _{\omega_{\theta,\phi}}\right|\\
 & \leq\left|\mathbb{E}_{p}\!\left[\frac{\partial}{\partial w_{\ell}}T_{w}(Z)\right]\right|+\left|\left\langle \sum_{z\in\mathcal{Z}}e^{T_{w}(z)}\frac{\partial}{\partial w_{\ell}}T_{w}(z)\Lambda_{z}\right\rangle _{\omega_{\theta,\phi}}\right|\\
 & \leq\left\Vert \frac{\partial}{\partial w_{\ell}}T_{w}\right\Vert +\left\Vert \sum_{z\in\mathcal{Z}}e^{T_{w}(z)}\frac{\partial}{\partial w_{\ell}}T_{w}(z)\Lambda_{z}\right\Vert \\
 & \leq\left\Vert \frac{\partial}{\partial w_{\ell}}T_{w}\right\Vert +\left\Vert e^{T_{w}}\frac{\partial}{\partial w_{\ell}}T_{w}\right\Vert \\
 & \leq\left\Vert \frac{\partial}{\partial w_{\ell}}T_{w}\right\Vert \left(\left\Vert e^{T_{w}}\right\Vert +1\right),
\end{align}
where the last inequality follows from submultiplicavity of the function
norm defined in \eqref{eq:observable-norm-bound} (i.e., $\left\Vert f\cdot g\right\Vert \leq\left\Vert f\right\Vert \left\Vert g\right\Vert $).
\end{proof}
\begin{prop}[Hessian bounds]
\label{prop:hessian-bounds}The following inequalities hold:
\begin{align}
\left\Vert H_{ww}\right\Vert ^{2} & \leq\sum_{\ell,m\in\left[L\right]}\left(\left\Vert \frac{\partial^{2}}{\partial w_{\ell}\partial w_{m}}T_{w}\right\Vert \left(1+\left\Vert e^{T_{w}}\right\Vert \right)+\left\Vert e^{T_{w}}\right\Vert \left\Vert \frac{\partial}{\partial w_{\ell}}T_{w}\right\Vert \left\Vert \frac{\partial}{\partial w_{m}}T_{w}\right\Vert \right)^{2},\label{eq:hessian-w-w-up-bnd}\\
\left\Vert H_{w\gamma}\right\Vert ^{2} & \leq4\left(\sum_{\ell\in\left[L\right]}\left\Vert e^{T_{w}}\frac{\partial}{\partial w_{\ell}}T_{w}\right\Vert ^{2}\right)\left(\sum_{j\in\left[J\right]}\left\Vert G_{j}\right\Vert ^{2}+\sum_{k\in\left[K\right]}\left\Vert H_{k}\right\Vert ^{2}\right).\label{eq:hessian-w-gamma-bound}
\end{align}
\end{prop}

\begin{proof}
Consider from \eqref{eq:hessian-w-w} that
\begin{align}
 & \left|H_{w_{\ell},w_{m}}\right|\nonumber \\
 & =\left|\begin{array}{c}
\mathbb{E}_{p}\!\left[\frac{\partial^{2}}{\partial w_{\ell}\partial w_{m}}T_{w}(Z)\right]\\
-\mathbb{E}_{q_{\gamma}}\!\left[e^{T_{w}(Z)}\left(\frac{\partial^{2}}{\partial w_{\ell}\partial w_{m}}T_{w}(Z)+\frac{\partial}{\partial w_{\ell}}T_{w}(Z)\frac{\partial}{\partial w_{m}}T_{w}(Z)\right)\right]
\end{array}\right|\\
 & \leq\left|\mathbb{E}_{p}\!\left[\frac{\partial^{2}}{\partial w_{\ell}\partial w_{m}}T_{w}(Z)\right]\right|\nonumber \\
 & \qquad+\left|\mathbb{E}_{q_{\gamma}}\!\left[e^{T_{w}(Z)}\left(\frac{\partial^{2}}{\partial w_{\ell}\partial w_{m}}T_{w}(Z)+\frac{\partial}{\partial w_{\ell}}T_{w}(Z)\frac{\partial}{\partial w_{m}}T_{w}(Z)\right)\right]\right|\\
 & \leq\left\Vert \frac{\partial^{2}}{\partial w_{\ell}\partial w_{m}}T_{w}\right\Vert +\left\Vert e^{T_{w}}\left(\frac{\partial^{2}}{\partial w_{\ell}\partial w_{m}}T_{w}+\frac{\partial}{\partial w_{\ell}}T_{w}\frac{\partial}{\partial w_{m}}T_{w}\right)\right\Vert \\
 & \leq\left\Vert \frac{\partial^{2}}{\partial w_{\ell}\partial w_{m}}T_{w}\right\Vert \left(1+\left\Vert e^{T_{w}}\right\Vert \right)+\left\Vert e^{T_{w}}\right\Vert \left\Vert \frac{\partial}{\partial w_{\ell}}T_{w}\right\Vert \left\Vert \frac{\partial}{\partial w_{m}}T_{w}\right\Vert .
\end{align}
The first bound in \eqref{eq:hessian-w-w-up-bnd} thus follows because
the spectral norm does not exceed the Schatten two-norm (i.e., Frobenius
norm). That is, $\left\Vert A\right\Vert ^{2}\leq\left\Vert A\right\Vert _{2}^{2}\coloneqq\sum_{i,j}\left|A_{i,j}\right|^{2}$.

Defining $P_{w_{\ell}}\coloneqq\sum_{z\in\mathcal{Z}}e^{T_{w}(z)}\frac{\partial}{\partial w_{\ell}}T_{w}(z)\Lambda_{z}$,
consider that
\begin{align}
 & \left|\frac{\partial}{\partial\theta_{j}}\frac{\partial}{\partial w_{\ell}}f(\gamma,w)\right|\nonumber \\
 & =\left|\frac{1}{2}\left\langle \left\{ e^{iH(\phi)}P_{w_{\ell}}e^{-iH(\phi)},\Phi_{\theta}(G_{j})\right\} \right\rangle _{\rho_{\theta}}-\left\langle P_{w_{\ell}}\right\rangle _{\omega_{\theta,\phi}}\left\langle G_{j}\right\rangle _{\rho_{\theta}}\right|\\
 & \leq\frac{1}{2}\left|\left\langle \left\{ e^{iH(\phi)}P_{w_{\ell}}e^{-iH(\phi)},\Phi_{\theta}(G_{j})\right\} \right\rangle _{\rho_{\theta}}\right|+\left|\left\langle P_{w_{\ell}}\right\rangle _{\omega_{\theta,\phi}}\left\langle G_{j}\right\rangle _{\rho_{\theta}}\right|\\
 & \leq\left\Vert e^{iH(\phi)}P_{w_{\ell}}e^{-iH(\phi)}\right\Vert \left\Vert \Phi_{\theta}(G_{j})\right\Vert +\left\Vert P_{w_{\ell}}\right\Vert \left\Vert G_{j}\right\Vert \\
 & \leq2\left\Vert P_{w_{\ell}}\right\Vert \left\Vert G_{j}\right\Vert \\
 & \leq2\left\Vert e^{T_{w}}\frac{\partial}{\partial w_{\ell}}T_{w}\right\Vert \left\Vert G_{j}\right\Vert .
\end{align}
Similarly, consider that
\begin{align}
\left|\frac{\partial}{\partial\phi_{k}}\frac{\partial}{\partial w_{\ell}}f(\gamma,w)\right| & =\left|i\left\langle \left[P_{w},\Psi_{\phi}(H_{k})\right]\right\rangle _{\omega_{\theta,\phi}}\right|\\
 & \leq2\left\Vert P_{w_{\ell}}\right\Vert \left\Vert \Psi_{\phi}(H_{k})\right\Vert \\
 & \leq2\left\Vert e^{T_{w}}\frac{\partial}{\partial w_{\ell}}T_{w}\right\Vert \left\Vert H_{k}\right\Vert .
\end{align}
Combining these bounds and noting that the spectral norm is bounded
from above by the Schatten two-norm (i.e., Frobenius norm), we conclude
the bound in \eqref{eq:hessian-w-gamma-bound}.
\end{proof}

\section{Hybrid quantum--classical algorithms for Born-rule generative modeling}

\label{sec:Hybrid-quantum=002013classical-algs}In this section, I
delineate several hybrid quantum--classical algorithms that can be
used for minimax optimization of the objective function $f(\gamma,w)$
in \eqref{eq:neural-net-param-div-obj-func}. These are directly based
on extragradient \cite{Korpelevich1976,Gorbunov2022}, two-timescale
gradient descent-ascent \cite{Lin2020,Lin2024}, follow-the-ridge
\cite{Wang2020}, and HessianFR \cite{Gao2022a}. Extragradient and
two-timescale gradient descent-ascent are purely first-order methods,
relying only on gradients. Follow-the-ridge and HessianFR are second-order,
leveraging Hessian information in addition to gradients in order to
navigate the local minimax geometry.

An important goal for any such minimax optimization algorithm is to
converge to a local minimax point \cite[Definition~14]{Jin2020} of
the objective function in \eqref{eq:neural-net-param-div-obj-func}.
On the one hand, the first two algorithms (extragradient and two-timescale
gradient descent-ascent) lack formal convergence guarantees to local
minimax points in general non-convex--non-concave settings, even
though they may perform well empirically. The latter two algorithms
(follow-the-ridge and HessianFR) do provide theoretical guarantees
for convergence to local minimax points (see \cite[Theorem~1]{Wang2020}
and \cite[Theorem~3.2]{Gao2022a}). On the other hand, the first two
algorithms require first-order information only (gradients), while
the latter two algorithms require both first- and second-order information
(gradients and Hessians) and thus are more computationally expensive.

As just mentioned, the goal of minimax optimization algorithms is
to converge to a local minimax point \cite[Definition~14]{Jin2020}
of \eqref{eq:neural-net-param-div-obj-func}. Recall that a point
$(\gamma^{\star},w^{\star})$ is a local minimax point if (i) both
partial gradients vanish at that point; (ii) $w^{\star}$ is a strict
local maximizer of $f(\gamma^{\star},\cdot)$; and (iii) $\gamma^{\star}$
is a local minimizer of the (local) value function $\max_{w}f(\gamma,w)$
obtained by maximizing $f(\gamma,w)$ over $w$ in a small neighborhood
of $w^{\star}$. In terms of curvature, when $f$ is twice differentiable,
condition (ii) is equivalent to the $ww$-block of the Hessian being
negative definite, $H_{ww}(\gamma^{\star},w^{\star})<0$, and condition
(iii) is equivalent to positive semi-definiteness of the associated
Schur complement:
\begin{equation}
H_{\gamma\gamma}(\gamma^{\star},w^{\star})-H_{\gamma w}(\gamma^{\star},w^{\star})\,H_{ww}(\gamma^{\star},w^{\star})^{-1}\,H_{w\gamma}(\gamma^{\star},w^{\star})\geq0.
\end{equation}
These geometric conditions motivate why some algorithms below explicitly
incorporate Hessian information and why they are able to guarantee
convergence to local minimax points. More explicitly, these curvature
conditions explain why methods such as follow-the-ridge and HessianFR
explicitly incorporate the matrix product $H_{ww}^{-1}H_{w\gamma}$:
they exploit the local concavity in $w$ and the induced curvature
in the $\gamma$-direction encoded by the Schur complement.

All hybrid quantum--classical algorithms discussed below require
estimating gradients, and some also require estimating Hessian blocks,
using a quantum device. These estimated quantities are then supplied
to classical routines that perform the update steps. In particular,
for second-order methods such as follow-the-ridge and HessianFR, the
classical computer assembles the quantum-estimated Hessian blocks
and computes matrix inverses such as $H_{ww}^{-1}$ numerically. As
a result, the stability of these updates depends both on the accuracy
of the quantum estimates and on the conditioning of $H_{ww}$; in
practice, regularization (e.g., $H_{ww}+\lambda I$) may be used to
ensure robust classical inversion. The choice of algorithm involves
a trade-off between quantum sampling cost and classical computational
overhead: first-order methods require fewer quantum resources, while
second-order methods demand more quantum circuit executions and matrix
inversions but can achieve guaranteed convergence.

In what follows, I assume that sufficiently many samples have been
taken when estimating both the gradients and Hessian elements, so
that the stochastic updates closely approximate the corresponding
deterministic updates. The bounds from Section \ref{subsec:Analytical-expressions-grad-Hess}
give a sense of the sample overhead needed to get good estimates of
the gradient and Hessian, and the bounds from Propositions \ref{prop:gradient-bounds}
and \ref{prop:hessian-bounds} give a sense of the number of iterations
needed to ensure convergence of the algorithms (see, e.g., \cite[Theorem~3.8]{Gao2022a}).

In the remainder of this section, I describe how each of these four
algorithms can be instantiated in the hybrid quantum--classical setting,
highlighting their computational requirements and the form of the
updates.

\subsection{Extragradient}

One of the simplest methods for minimax optimization is the extragradient
method \cite{Korpelevich1976,Gorbunov2022}. Going beyond the standard
gradient descent-ascent approach to minimax optimization, it involves
a predictor and corrector step. 
\begin{lyxalgorithm}
\label{alg:Extragradient}Extragradient proceeds according to the
following steps:
\begin{enumerate}
\item Inputs: learning rates $\eta_{\gamma},\eta_{w}>0$ and number of iterations
$N$:
\item For $m=1,\ldots,N$
\begin{enumerate}
\item $\widetilde{\gamma}_{m}\leftarrow\gamma_{m}-\eta_{\gamma}\nabla_{\gamma}f(\gamma_{m},w_{m})$,
where $\nabla_{\gamma}f$ is given in \eqref{eq:gradient-gamma-theta}
and \eqref{eq:gradient-gamma-phi}.
\item $\widetilde{w}_{m}\leftarrow w_{m}+\eta_{w}\nabla_{w}f(\gamma_{m},w_{m})$,
where $\nabla_{w}f$ is given in \eqref{eq:gradient-w}.
\item $\gamma_{m+1}\leftarrow\gamma_{m}-\eta_{\gamma}\nabla_{\gamma}f(\widetilde{\gamma}_{m},\widetilde{w}_{m})$,
where $\nabla_{\gamma}f$ is given in \eqref{eq:gradient-gamma-theta}
and \eqref{eq:gradient-gamma-phi}.
\item $w_{m+1}\leftarrow w_{m}+\eta_{w}\nabla_{w}f(\widetilde{\gamma}_{m},\widetilde{w}_{m})$,
where $\nabla_{w}f$ is given in \eqref{eq:gradient-w}.
\end{enumerate}
\end{enumerate}
\end{lyxalgorithm}

\subsection{Two-timescale gradient descent-ascent}

Another method for minimax optimization, for which the goal is to
converge to a local minimax point, is two-timescale gradient descent-ascent
\cite{Lin2020,Lin2024}. The main idea of this approach is to choose
the step sizes for the descent update and ascent update to change
on different timescales, so that the ascent updates are much larger
/ faster than the descent updates (i.e., $\eta_{\gamma}\ll\eta_{w}$).
See \cite{Lin2024} for formal guarantees and specific methods for
choosing the step sizes $\eta_{\gamma}$ and $\eta_{w}$.
\begin{lyxalgorithm}
\label{alg:Two-timescale-gda}Two-timescale gradient descent-ascent
proceeds according to the following steps:
\begin{enumerate}
\item Inputs: learning rates $\eta_{\gamma},\eta_{w}>0$, where $\eta_{\gamma}\ll\eta_{w}$
and number of iterations $N$:
\item For $m=1,\ldots,N$
\begin{enumerate}
\item $\gamma_{m}\leftarrow\gamma_{m}-\eta_{\gamma}\nabla_{\gamma}f(\gamma_{m},w_{m})$,
where $\nabla_{\gamma}f$ is given in \eqref{eq:gradient-gamma-theta}
and \eqref{eq:gradient-gamma-phi}.
\item $w_{m}\leftarrow w_{m}+\eta_{w}\nabla_{w}f(\gamma_{m},w_{m})$, where
$\nabla_{w}f$ is given in \eqref{eq:gradient-w}.
\end{enumerate}
\end{enumerate}
\end{lyxalgorithm}

\subsection{Follow-the-ridge}

Follow-the-ridge is a second-order algorithm for minimax optimization
\cite{Wang2020}, which provably converges to a local minimax point
\cite{Jin2020} of the objective function in \eqref{eq:neural-net-param-div-obj-func}
(see \cite[Theorem~1]{Wang2020}). Thus, we can use it for optimizing
\eqref{eq:neural-net-param-div-obj-func}. See \cite[Appendix~A]{Wang2020}
for conditions on the learning rates needed to guarantee local convergence.
See also \cite[Section~4.1]{Wang2020} for variants of the basic algorithm
that improve convergence, involving preconditioning and momentum.
\begin{lyxalgorithm}
\label{alg:Follow-the-ridge}Follow-the-ridge proceeds according to
the following steps:
\begin{enumerate}
\item Inputs: learning rates $\eta_{\gamma},\eta_{w}>0$ and number of iterations
$N$:
\item For $m=1,\ldots,N$
\begin{enumerate}
\item $\gamma_{m+1}\leftarrow\gamma_{m}-\eta_{\gamma}\nabla_{\gamma}f(\gamma_{m},w_{m})$,
where $\nabla_{\gamma}f$ is given in \eqref{eq:gradient-gamma-theta}
and \eqref{eq:gradient-gamma-phi}.
\item $w_{m+1}\leftarrow w_{m}+\eta_{w}\nabla_{w}f(\gamma_{m},w_{m})+\eta_{\gamma}H_{ww}^{-1}H_{w\gamma}\nabla_{\gamma}f(\gamma_{m},w_{m})$,
where $\nabla_{w}f$ is given in \eqref{eq:gradient-w}, $H_{ww}$
in \eqref{eq:hessian-w-w}, and $H_{w\gamma}$ in \eqref{eq:hessian-w-gamma-1}
and \eqref{eq:hessian-w-gamma-2}.
\end{enumerate}
\end{enumerate}
\end{lyxalgorithm}

Follow-the-ridge incorporates the term $H_{ww}^{-1}H_{w\gamma}\nabla_{\gamma}f$,
which corrects the gradient update along the ascent variable direction,
ensuring that the updates remain aligned with the local ridge structure
of the minimax landscape.

\subsection{HessianFR}

The final algorithm that I consider is HessianFR \cite{Gao2022a},
which extends follow-the-ridge with additional Newton-like corrections
designed to accelerate convergence and improve local stability. As
shown in \cite[Theorem~3.2]{Gao2022a}, it provably converges to a
local minimax point.
\begin{lyxalgorithm}
\label{alg:HessianFR}HessianFR proceeds according to the following
steps:
\begin{enumerate}
\item Inputs: learning rates $\eta_{\gamma},\eta_{w1},\eta_{w2}>0$ and
number of iterations $N$:
\item For $m=1,\ldots,N$
\begin{enumerate}
\item $\gamma_{m+1}\leftarrow\gamma_{m}-\eta_{\gamma}\nabla_{\gamma}f(\gamma_{m},w_{m})$,
where $\nabla_{\gamma}f$ is given in \eqref{eq:gradient-gamma-theta}
and \eqref{eq:gradient-gamma-phi}.
\item $w_{m+1}\leftarrow w_{m}+u_{m}$, where
\begin{multline}
u_{m}\equiv\eta_{w1}\nabla_{w}f(\gamma_{m},w_{m})-\eta_{w2}H_{ww}^{-1}\nabla_{w}f(\gamma_{m},w_{m})\\
+\eta_{\gamma}H_{ww}^{-1}H_{w\gamma}\nabla_{\gamma}f(\gamma_{m},w_{m}),
\end{multline}
$\nabla_{w}f$ is given in \eqref{eq:gradient-w}, $H_{ww}$ in \eqref{eq:hessian-w-w},
and $H_{w\gamma}$ in \eqref{eq:hessian-w-gamma-1} and \eqref{eq:hessian-w-gamma-2}.
\end{enumerate}
\end{enumerate}
\end{lyxalgorithm}

Taken together, all four hybrid quantum--classical approaches presented
in this section span a spectrum from purely first-order updates to
fully second-order corrections, providing a flexible toolkit for optimizing
\eqref{eq:neural-net-param-div-obj-func} under different quantum
resource and accuracy constraints.

\section{Maintaining concavity via a linear model in feature space}

\label{sec:Maintaining-concavity-via-MLFS}A key advantage of the
DV objective function in \eqref{eq:DV-formula} is that it is concave
in $T$. This is a desirable feature for optimization, so that optimizing
\eqref{eq:min-max-opt-rel-ent} directly is a non-convex--concave
optimization problem, for which there are many more guarantees than
a non-convex--non-concave optimization problem (see \cite{Lin2024}
and references therein). As such, it could be desirable to maintain
concavity when parameterizing the maximization problem, and one way
to do so is to employ a linear model in feature space \cite{ShalevShwartz2014},
rather than a general neural network. The trade-off when employing
a linear model in feature space is that it may not be as expressive
as a general neural network.

In this section, I derive analytical expressions for the gradient
and Hessian when using a linear model in feature space. These expressions
can be estimated on a quantum computer using Algorithms \ref{alg:EQBGE-1}
and \ref{alg:EQBGE-2}, and the estimates can be plugged directly
into Algorithms \ref{alg:Extragradient}--\ref{alg:HessianFR} for
minimax optimization. In particular, let us pick
\begin{equation}
T_{w}(z)\coloneqq w^{T}\zeta(z),\label{eq:linear-model-feature-space}
\end{equation}
where $w\in\mathbb{R}^{L}$ and $\zeta\colon\mathcal{Z}\to\mathbb{R}^{L}$
is a feature map and can be nonlinear. The function in \eqref{eq:linear-model-feature-space}
is known as a linear model in feature space because it is linear in
the parameter vector $w$. With this choice, the objective function
$f(\gamma,w)$ in \eqref{eq:neural-net-param-div-obj-func} becomes
as follows:
\begin{equation}
f(\gamma,w)=\mathbb{E}_{p}\!\left[w^{T}\zeta(Z)\right]+1-\mathbb{E}_{q_{\gamma}}\!\left[e^{w^{T}\zeta(Z)}\right].
\end{equation}
As such, it is concave in $w$ because the first term $\mathbb{E}_{p}\!\left[w^{T}\zeta(Z)\right]$
is linear in $w$ and the last term $\mathbb{E}_{q_{\gamma}}\!\left[e^{w^{T}\zeta(Z)}\right]$
is convex in $w$, due to the exponential function being convex, so
that $-\mathbb{E}_{q_{\gamma}}\!\left[e^{w^{T}\zeta(Z)}\right]$ is
concave in $w$. With this choice, the optimization problem in \eqref{eq:neural-net-param-div}
is indeed a non-convex--concave optimization problem.

We can modify the objective function just slightly, by means of a
quadratic regularization, as follows:
\begin{equation}
\widetilde{f}(\gamma,w)\coloneqq\mathbb{E}_{p}\!\left[w^{T}\zeta(Z)\right]+1-\mathbb{E}_{q_{\gamma}}\!\left[e^{w^{T}\zeta(Z)}\right]-\frac{\lambda}{2}\left\Vert w\right\Vert ^{2},\label{eq:modified-obj-strongly-concave}
\end{equation}
where $\lambda>0$ is a regularization parameter. Doing so just slightly
changes the optimization problem while providing stronger numerical
stability guarantees for it. Indeed, the modified objective function
$\widetilde{f}(\gamma,w)$ is strongly concave in $w$, given that
$f(\gamma,w)$ is concave in $w$ \cite[Lemma~2.12]{Garrigos2024}.
The modified optimization problem is as follows:
\begin{equation}
\inf_{\gamma\in\Gamma}\sup_{w\in\mathcal{W}}\widetilde{f}(\gamma,w),\label{eq:modified-opt-strongly-concave}
\end{equation}
and the following inequalities clearly hold:
\begin{equation}
\inf_{\gamma\in\Gamma}D(p\|q_{\gamma})\geq\inf_{\gamma\in\Gamma}\sup_{w\in\mathcal{W}}f(\gamma,w)\geq\inf_{\gamma\in\Gamma}\sup_{w\in\mathcal{W}}\widetilde{f}(\gamma,w).
\end{equation}
One recovers the original minimax optimization $\inf_{\gamma\in\Gamma}\sup_{w\in\mathcal{W}}f(\gamma,w)$
simply by setting $\lambda=0$.

The gradient of $\widetilde{f}(\gamma,w)$ with respect to $\gamma=\left(\theta,\phi\right)$
is the same as given in \eqref{eq:gradient-gamma-theta} and \eqref{eq:gradient-gamma-phi},
with $O_{w}$ therein set to $O_{w}=\sum_{z\in\mathcal{Z}}e^{w^{T}\zeta(z)}\Lambda_{z}$.
In addition, the number of samples sufficient for obtaining $\varepsilon$-accurate
estimates with error probability $\leq\delta$ is given by $\mathcal{O}\!\left(\frac{\left\Vert h\right\Vert ^{2}}{\varepsilon^{2}}\ln\!\left(\frac{1}{\delta}\right)\right)$,
where $h(z)\coloneqq e^{w^{T}\zeta(z)}$.

The gradient of $\widetilde{f}(\gamma,w)$ with respect to $w$ is
as follows:
\begin{equation}
\nabla_{w}\widetilde{f}(\gamma,w)=\mathbb{E}_{p}\!\left[\zeta(Z)\right]-\mathbb{E}_{q_{\gamma}}\!\left[e^{w^{T}\zeta(Z)}\zeta(Z)\right]-\lambda w,\label{eq:gradient-w-1}
\end{equation}
and the Hessian is given by
\begin{equation}
H_{ww}\coloneqq\nabla_{w,w}\widetilde{f}(\gamma,w)=-\mathbb{E}_{q_{\gamma}}\!\left[e^{w^{T}\zeta(Z)}\zeta(Z)\zeta(Z)^{T}\right]-\lambda I.\label{eq:hessian-w-w-1}
\end{equation}
Given that the matrix $\mathbb{E}_{q_{\gamma}}\!\left[e^{w^{T}\zeta(Z)}\zeta(Z)\zeta(Z)^{T}\right]$
is positive semi-definite, the extra term $\lambda I$ with $\lambda>0$
guarantees that $H_{ww}$ is a negative definite matrix. Both the
gradient in \eqref{eq:gradient-w-1} and the Hessian in \eqref{eq:hessian-w-w-1}
can be estimated by sampling from $p$ and $q_{\gamma}$.

By following a development similar to that in \eqref{eq:rewrite-as-obs-O-1}--\eqref{eq:rewrite-as-obs-O-last},
we can rewrite the components of the gradient $\nabla_{w}\widetilde{f}(\gamma,w)$
in \eqref{eq:gradient-w-1} as follows:
\begin{equation}
\frac{\partial}{\partial w_{\ell}}\widetilde{f}(\gamma,w)=\mathbb{E}_{p}\!\left[\zeta_{\ell}(Z)\right]-\left\langle P_{w_{\ell}}\right\rangle _{\omega_{\theta,\phi}}-\lambda w_{\ell},\label{eq:LMFS-gradient-w}
\end{equation}
where the observable $P_{w_{\ell}}$ is defined as
\begin{equation}
P_{w_{\ell}}\coloneqq\sum_{z\in\mathcal{Z}}e^{w^{T}\zeta(z)}\zeta_{\ell}(z)\Lambda_{z},
\end{equation}
with $\zeta_{\ell}(z)$ the $\ell$th output of $\zeta(z)$. It then
follows as a direct consequence of \eqref{eq:gradient-VQE-obj}--\eqref{eq:time-evolve-grad}
that the Hessian elements of $\widetilde{f}(\gamma,w)$ with respect
to $\gamma=\left(\theta,\phi\right)$ and $w$, denoted collectively
by $H_{\gamma w}$, are as follows:
\begin{align}
\frac{\partial}{\partial\theta_{j}}\frac{\partial}{\partial w_{\ell}}\widetilde{f}(\gamma,w) & =\frac{\partial}{\partial\theta_{j}}\left(-\left\langle P_{w_{\ell}}\right\rangle _{\omega_{\theta,\phi}}\right)\\
 & =\frac{1}{2}\left\langle \left\{ e^{iH(\phi)}P_{w_{\ell}}e^{-iH(\phi)},\Phi_{\theta}(G_{j})\right\} \right\rangle _{\rho_{\theta}}-\left\langle P_{w_{\ell}}\right\rangle _{\omega_{\theta,\phi}}\left\langle G_{j}\right\rangle _{\rho_{\theta}},\label{eq:LMFS-hessian-w-theta}\\
\frac{\partial}{\partial\phi_{k}}\frac{\partial}{\partial w_{\ell}}\widetilde{f}(\gamma,w) & =\frac{\partial}{\partial\phi_{k}}\left(-\left\langle P_{w_{\ell}}\right\rangle _{\omega_{\theta,\phi}}\right)\\
 & =i\left\langle \left[P_{w_{\ell}},\Psi_{\phi}(H_{k})\right]\right\rangle _{\omega_{\theta,\phi}}.\label{eq:LMFS-hessian-w-phi}
\end{align}
Thus, each of the terms in \eqref{eq:LMFS-gradient-w}, \eqref{eq:LMFS-hessian-w-theta},
and \eqref{eq:LMFS-hessian-w-phi} can again be estimated by means
of Algorithms \ref{alg:EQBGE-1} and \ref{alg:EQBGE-2}, with $O$
therein replaced by $P_{w_{\ell}}$. In addition, the number of samples
sufficient for obtaining $\varepsilon$-accurate estimates with error
probability $\leq\delta$ is given by $\mathcal{O}\!\left(\frac{\left\Vert p_{\ell}\right\Vert ^{2}}{\varepsilon^{2}}\ln\!\left(\frac{1}{\delta}\right)\right)$,
where $p_{\ell}(z)\coloneqq e^{w^{T}\zeta(z)}\zeta_{\ell}(z)$.

\section{Alternative distinguishability measures}

\label{sec:Alternative-dist-meas}In this section, I note that the
main method outlined in this paper applies to other distinguishability
measures. This holds because these other distinguishability measures
have variational representations analogous to the Donsker--Varadhan
formula in \eqref{eq:DV-formula}. As such, one can extend the whole
development to these other distinguishability measures. See, e.g,
\cite[Section~2.2]{Sreekumar2022} for variational representations
of other distinguishability measures like chi-square divergence, squared
Hellinger distance, and total variation distance.

As a key example, let us consider the R\'enyi relative quasi-entropy,
defined for probability distributions $p$ and $q_{\gamma}$ and $\alpha\in\left(0,1\right)\cup\left(1,\infty\right)$
as
\begin{equation}
Q_{\alpha}(p\|q_{\gamma})\coloneqq\sum_{z\in\mathcal{\mathcal{Z}}}p(z)^{\alpha}q_{\gamma}(z)^{1-\alpha}.
\end{equation}
It is known to have the following variational representations \cite{Berta2017}:
\begin{align}
\inf_{T\in\mathcal{T}}\left\{ \alpha\mathbb{E}_{p}\!\left[e^{\frac{\alpha-1}{\alpha}T(Z)}\right]+\left(1-\alpha\right)\mathbb{E}_{q_{\gamma}}\!\left[e^{T(Z)}\right]\right\}  & \quad\text{for }\alpha\in\left(0,1\right),\\
\sup_{T\in\mathcal{T}}\left\{ \alpha\mathbb{E}_{p}\!\left[e^{\frac{\alpha-1}{\alpha}T(Z)}\right]+\left(1-\alpha\right)\mathbb{E}_{q_{\gamma}}\!\left[e^{T(Z)}\right]\right\}  & \quad\text{for }\alpha\in\left(1,\infty\right),
\end{align}
which results from simplifying the expressions in \cite[Lemma~3]{Berta2017}
to the classical case. Thus, for $\alpha\in\left(0,1\right)$, the
parameterized maximin optimization problem for Born-rule generative
modeling becomes
\begin{align}
\sup_{\gamma\in\Gamma}Q_{\alpha}(p\|q_{\gamma}) & \leq\sup_{\gamma\in\Gamma}\inf_{w\in\mathcal{W}}f_{\alpha}(\gamma,w),\label{eq:neural-net-param-div-renyi-0-1}
\end{align}
and for $\alpha>1$, it is as follows:
\begin{align}
\inf_{\gamma\in\Gamma}Q_{\alpha}(p\|q_{\gamma}) & \geq\inf_{\gamma\in\Gamma}\sup_{w\in\mathcal{W}}f_{\alpha}(\gamma,w),\label{eq:neural-net-param-div-renyi-gt-1}
\end{align}
where
\begin{equation}
f_{\alpha}(\gamma,w)\coloneqq\alpha\mathbb{E}_{p}\!\left[e^{\frac{\alpha-1}{\alpha}T_{w}(Z)}\right]+\left(1-\alpha\right)\mathbb{E}_{q_{\gamma}}\!\left[e^{T_{w}(Z)}\right].
\end{equation}
Let us note that the optimizations in \eqref{eq:neural-net-param-div-renyi-0-1}
are flipped when compared to \eqref{eq:neural-net-param-div} because
$Q_{\alpha}$ is a similarity measure for $\alpha\in(0,1)$, taking
values in $\left[0,1\right]$ with its largest value equal to one
if and only if $p=q_{\gamma}$. When $\alpha>1$, the quantity $Q_{\alpha}$
is a distinguishability measure, taking values $\geq1$ with its minimum
occurring at one, if and only if $p=q_{\gamma}$.

The gradient $\nabla_{\gamma}f_{\alpha}(\gamma,w)$ with respect to
$\gamma=\left(\theta,\phi\right)$ is precisely the same as in \eqref{eq:gradient-gamma-theta}
and \eqref{eq:gradient-gamma-phi}, with the exception that it should
be multiplied by the prefactor $1-\alpha$. The gradient $\nabla_{w}f_{\alpha}(\gamma,w)$
and Hessian $\nabla_{w}^{2}f_{\alpha}(\gamma,w)$ have the following
expressions:
\begin{multline}
\nabla_{w}f_{\alpha}(\gamma,w)=\\
\left(\alpha-1\right)\left(\mathbb{E}_{p}\!\left[e^{\frac{\alpha-1}{\alpha}T_{w}(Z)}\nabla_{w}T_{w}(Z)\right]-\mathbb{E}_{q_{\gamma}}\!\left[e^{T_{w}(Z)}\nabla_{w}T_{w}(Z)\right]\right),
\end{multline}
\begin{multline}
\nabla_{w}^{2}f_{\alpha}(\gamma,w)=\\
\left(\alpha-1\right)\mathbb{E}_{p}\!\left[e^{\frac{\alpha-1}{\alpha}T_{w}(Z)}\left(\nabla_{w}^{2}T_{w}(Z)+\left(\frac{\alpha-1}{\alpha}\right)\nabla_{w}T_{w}(Z)\left[\nabla_{w}T_{w}(Z)\right]^{T}\right)\right]\\
+\left(1-\alpha\right)\mathbb{E}_{q_{\gamma}}\!\left[e^{T_{w}(Z)}\left(\nabla_{w}^{2}T_{w}(Z)+\nabla_{w}T_{w}(Z)\left[\nabla_{w}T_{w}(Z)\right]^{T}\right)\right],
\end{multline}
and thus can be estimated by sampling from $p$ and $q_{\gamma}$
and performing the standard backpropagation algorithm. Finally, the
components of the Hessian $H_{w\gamma}$ are precisely the same as
in \eqref{eq:hessian-w-gamma-1} and \eqref{eq:hessian-w-gamma-2},
with the exception that they should be multiplied by the prefactor
$1-\alpha$.

These expressions can be directly plugged into Algorithms \ref{alg:Extragradient}--\ref{alg:HessianFR}
for conducting the minimax optimization for \eqref{eq:neural-net-param-div-renyi-gt-1},
while noting that the signs of the updates should be flipped when
conducting the maximin optimization for \eqref{eq:neural-net-param-div-renyi-0-1}.

\section{Conclusion}

\label{sec:Conclusion}In summary, I delineated a practical solution
to the problem of training evolved quantum Boltzmann machines for
Born-rule generative modeling. The solution combines the evolved quantum
Boltzmann gradient estimator of \cite{Patel2024,Minervini2025} and
the Donsker--Vardhan variational formula for the relative entropy,
leading to four different hybrid quantum--classical algorithms for
training evolved quantum Boltzmann machines. As a special case, this
provides a practical solution to the open problem, dating back to
\cite{Amin2018}, of training quantum Boltzmann machines for Born-rule
generative modeling. I also showed how the method extends to other
distinguishability measures like the R\'enyi relative quasi-entropy.

Going forward from here, there are a number of open directions to
pursue. First, following \cite{Sreekumar2022}, one could aim to provide
formal guarantees for the convergence of the relative entropy estimator.
This would involve separating the total error into three components:
(i) an approximation error, reflecting how well the neural network
approximates the optimal value in the Donsker--Varadhan formula,
(ii) an estimation error, quantifying how quickly the estimator converges
to its expected value, and (iii) an optimization error, measuring
the convergence of the minimax optimization to a local minimax point.
Second, it would be worthwhile to perform numerical simulations to
assess practical performance. Here, one could employ classical simulators
of quantum computers, or one could execute the algorithms on real
hardware. Finally, it would be interesting to investigate whether
these hybrid quantum--classical algorithms achieve improved convergence
when combined with the natural gradient method of \cite{Patel2025}.

\medskip{}

\textit{Acknowledgements}---I acknowledge insightful discussions
with Dhrumil Patel and Michele Minervini during the development of
\cite{Patel2024,Minervini2025}, and additional insightful discussions
with Ziv Goldfeld and Sreejith Sreekumar during the development of
\cite{Sreekumar2025}. I also acknowledge support from the National
Science Foundation under grant no.~2329662 and from the Cornell School
of Electrical and Computer Engineering.

\footnotesize

\bibliographystyle{alphaurl}
\bibliography{Ref}

\end{document}